\documentclass[journal,onecolumn]{IEEEtran}

\usepackage{enumerate}
\usepackage[shortlabels]{enumitem}

\usepackage{amssymb}
\usepackage{amsfonts}
\usepackage{mathrsfs}
\usepackage{amsmath} 
\usepackage{amsthm}
\usepackage{mathtools}

\newcommand{\ceil}[1]{ {\lceil#1\rceil}}
\newcommand{\floor}[1]{ {\lfloor#1\rfloor}}
\newcommand{\bi}{{\{0,1\}}}

\usepackage{algorithm}
\usepackage[noend]{algpseudocode}
\algrenewcommand\algorithmiccomment[1]{\textcolor{lightgray}{\hfill // #1}}

\usepackage{pgfplots}
\pgfplotsset{compat=1.18}
\usepackage{graphicx}

\usepackage[style=ieee,
            backend=biber,
            doi=false,
            maxbibnames=99,
            ]{biblatex}
 
\usepackage{subcaption}
\usepackage{wrapfig}

\theoremstyle{definition}
\newtheorem{theorem}{Theorem}[section]
\newtheorem{lemma}[theorem]{Lemma}
\newtheorem{corollary}[theorem]{Corollary}
\newtheorem{remark}[theorem]{Remark}

\newtheorem{definition}[theorem]{Definition}

\numberwithin{comment}{section}

\renewcommand{\thetheorem}{\arabic{section}.\arabic{theorem}}
\usepackage{xcolor}
\usepackage{soul} 

\usepackage{hyperref}
\hypersetup{
 colorlinks,
 linkcolor={blue!100!black},
 citecolor={blue!100!black},
 urlcolor={blue!80!black} 
}

\newcommand{\bbF}{\mathbb{F}}

\newcommand{\bfb}{\mathbf{b}}
\newcommand{\bfc}{\mathbf{c}}
\newcommand{\bfd}{\mathbf{d}}

\newcommand{\bff}{\mathbf{f}}

\newcommand{\bfm}{\mathbf{m}}
\newcommand{\bfn}{\mathbf{n}}

\newcommand{\bfs}{\mathbf{s}}

\newcommand{\bfu}{\mathbf{u}}
\newcommand{\bfv}{\mathbf{v}}
\newcommand{\bfw}{\mathbf{w}}
\newcommand{\bfx}{\mathbf{x}}
\newcommand{\bfy}{\mathbf{y}}
\newcommand{\bfz}{\mathbf{z}}

\newcommand{\cB}{\mathcal{B}}
\newcommand{\cC}{\mathcal{C}}

\newcommand{\cF}{\mathcal{F}}
\newcommand{\cG}{\mathcal{G}}

\newcommand{\cI}{\mathcal{I}}

\newcommand{\cP}{\mathcal{P}}

\newcommand{\cS}{\mathcal{S}}

\newcommand{\cW}{\mathcal{W}}

\newcommand{\brc}{\mathrm{BRC}}
\newcommand{\muc}{\mathrm{MU}}
\newcommand{\rep}{\mathrm{rep}}

\begin{document}

\title{Optimal Break-Resilient Codes}

\author{Canran Wang,~\IEEEmembership{Member,~IEEE}%
}

\maketitle

\begin{abstract}

Break-resilient codes enable reliable communication in the presence of an omniscient adversary that may split a transmitted message at arbitrary boundaries between consecutive symbols, while the receiver observes only an unordered multiset of the resulting fragments.
For binary codewords of length~$n$ subject to at most~$t$ breaks, the best known explicit construction has redundancy~$O(t\log_2 n\log_2\log_2\log_2 n)$, whereas the information-theoretic lower bound is~$\Omega(t\log_2(n/t))$.
In this paper, we extend the binary break model to any fixed finite field~$\bbF_q$ and establish a redundancy lower bound of~$\Omega(t\log_q(n/t))$.
We then give an explicit construction of~$q$-ary break-resilient codes with redundancy~$O(t\log_q n)$ when~$t\leq n^{1-\varepsilon}$ for a fixed constant~$\varepsilon\in(0,1)$, matching the information-theoretic lower bound up to a constant factor.
The key idea is to compute a short algebraic fingerprint of the message, which enables the decoder to reject incorrect assemblies of the received fragments.
\end{abstract}

\section{Introduction}

An~$(n,t)$-break-resilient code ($(n,t)$-BRC) is a collection of length-$n$ codewords, each of which can be uniquely recovered when an adversary makes at most~$t$ breaks between adjacent symbols, and the decoder receives only the resulting multiset of at most~$t+1$ fragments.
Wang et al.~\cite{wang2026break} introduced this model and established an~$\Omega(t\log_2 (n/t))$ lower bound over the binary alphabet, together with an explicit binary construction of redundancy~$O(t\log_2 n\log_2\log_2 n)$.
The deterministic block-edit codes of Cheng et al.~\cite{cheng2019block} also imply an explicit binary BRC construction: an unordered collection of at most~$t+1$ fragments can be viewed as a permutation of at most~$t+1$ contiguous blocks, which can be realized by~$O(t)$ block transpositions.
Their construction therefore gives redundancy~$O(t\log_2 n\log_2\log_2\log_2 n)$.

Yet, both constructions assume a binary alphabet and leave a gap to the~$\Omega(t\log_2 (n/t))$ lower bound.
The reason is structural.
Both constructions achieve synchronization by decomposing the recovery task into multiple dependent stages, and the synchronization information required at each stage is protected separately, rather than being protected once globally.
Hence, although the adversary has only a ``budget'' of~$t$ breaks, each stage must be prepared for the possibility that all~$t$ breaks affect the information needed at that stage.
In this sense, the same adversarial budget is effectively paid for repeatedly across the recovery procedure, leading to the extra logarithmic factors in the redundancy.

In this paper, we generalize the binary break model to any fixed finite field~$\bbF_q$, and develop an alternative code construction that avoids the aforementioned staged decoding routine.
In particular, we first extend the information-theoretic lower bound of redundancy from~$\Omega(t\log_2(n/t))$ over~$\bi$ to~$\Omega(t\log_q(n/t))$ over~$\bbF_q$.
Then, we present an~$(n,t)$-BRC with redundancy~$O(t\log_q n)$, which achieves the lower bound throughout the regime~$t\leq n^{1-\varepsilon}$ for a fixed~$\varepsilon\in(0,1)$.

At a high level, the encoder computes a short fingerprint of the message, protects it using an MDS code, and places the encoded header at the start of the codeword.
The MDS protection enables the decoder to recover the fingerprint value from the fragmented codeword.
The recovered fingerprint is then used to identify the unique valid ordering of the received fragments.
An ordering is retained only if its concatenation begins with the MDS encoding of the recovered fingerprint, and the suffix following this prefix is called a~\emph{candidate}.
The fingerprinting function is carefully designed so that, among all candidates, the transmitted message is the unique one whose fingerprint equals the recovered fingerprint value.
The decoder can therefore identify the transmitted message without ambiguity.

\subsection*{Related Work}

Coding for unordered fragments has been studied under several models, many of which are motivated by applications in DNA storage.
In the \emph{sliced-channel} model, a codeword is partitioned at prescribed, evenly spaced locations, producing an unordered collection of equal-length substrings~\cite{sima2021coding,sima2024robust,lenz2019coding}.
The~\emph{torn-paper channel} instead places cuts according to a probabilistic process, resulting in unordered fragments of random lengths~\cite{shomorony2021torn,shomorony2020communicating,liu2026improved}.
An adversarial counterpart was studied in~\cite{bar2023adversarial} under the restriction that fragment lengths lie between prescribed lower and upper bounds.

In contrast, break-resilient codes~\cite{wang2026break} constrain only the number of adversarial breaks and place no restrictions on the lengths of the resulting fragments.
The model has been applied to robust information reconstruction in 3D-printed objects for forensic applications~\cite{wang2025secure}.
A subsequent extension, known as~$(t,s)$-break-resilient codes~\cite{wang2026breakloss}, additionally tolerates the complete loss of any subset of fragments whose aggregate length is at most~$s$.

\subsection*{Our Contributions}

Our contributions are twofold.
First, we extend the information-theoretic lower-bound argument for binary break-resilient codes to an arbitrary fixed finite field~$\bbF_q$, establishing a redundancy lower bound of~$\Omega(t\log_q(n/t))$ for~$q$-ary codes.
Second, for every fixed~$\varepsilon\in(0,1)$ and~$t\leq n^{1-\varepsilon}$, we give a deterministic construction of~$q$-ary break-resilient codes with redundancy~$O(t\log_q n)$.
This construction is order-optimal with the explicit redundancy
\begin{equation*}
   (6t+8)\ceil{\log_q k}+(3t+4)\cdot (\ceil{\log_q t}+11)+1 = O(t\log_q n).
\end{equation*}
The algebraic fingerprint underlying the construction can be viewed as an instance of the one-round cover-free recoloring framework of~\cite{li2026constructing}, specialized to the confusion graph induced by the break model.

\section{Model and Notation}\label{section:notation}

We use standard notation for strings.
For a positive integer~$n$ and an alphabet~$\Sigma_q$, a length-$n$ string over~$\Sigma_q$ is denoted by~$\bfc=(c_1,\ldots,c_n)\in\Sigma_q^n$.
We write~$|\bfc|$ for the length of~$\bfc$ and use~$\circ$ to denote string concatenation.
For~$n\in\mathbb N$, let~$[n]\triangleq\{1,2,\ldots,n\}$.
For an interval~$I=\{i,i+1,\ldots,j\}\subseteq[n]$, we denote by~$\bfx_I$ the substring~$\bfx_I=(x_i,x_{i+1},\ldots,x_j)$.
Throughout this paper, all logarithms are base~$2$ unless otherwise stated.

We next formally define the~\emph{$t$-break channel}.
Fix integers~$n\ge 1$ and~$0\le t\le n-1$,
and a~\emph{break pattern} is a set
\begin{equation*}
  \cP=\{b_1,\ldots,b_\tau\}\subseteq[n-1],
  0\le\tau\le t,
  b_1<\cdots<b_\tau .
\end{equation*}
Set~$b_0=0$ and~$b_{\tau+1}=n$, and define
\begin{equation*}
  I_{\cP,i}
  =
  \{b_i+1,b_i+2,\ldots,b_{i+1}\},
  \qquad
  i\in\{0,\ldots,\tau\}.
\end{equation*}
Thus,~$I_{\cP,0},\ldots,I_{\cP,\tau}$ are the consecutive nonempty intervals induced by the break pattern~$\cP$.

For an input string~$\bfc\in\Sigma_q^n$, the channel output corresponding to~$\cP$ is the unordered multiset
\begin{equation*}
  \mathsf{Frag}_{\cP}(\bfc)
  \triangleq
  \{\{
    \bfc_{I_{\cP,0}},
    \bfc_{I_{\cP,1}},
    \ldots,
    \bfc_{I_{\cP,\tau}}
  \}\},
\end{equation*}
where~$\{\{\cdot\}\}$ denotes a multiset.
The set of all possible outputs of the~$t$-break channel on input
$\bfc$ is
\begin{equation*}
  \cF_t(\bfc)
  \triangleq
  \left\{
    \mathsf{Frag}_{\cP}(\bfc):
    \cP\subseteq[n-1],\ |\cP|\le t
  \right\}.
\end{equation*}
We further define the channel output space as
\begin{equation*}
  \textstyle\cF_{n,t}\triangleq\bigcup_{\bfc\in\Sigma_q^n}\cF_t(\bfc).
\end{equation*}

A~$q$-ary~$(n,t)$-BRC carrying~$k$ information symbols consists of an encoder and a decoder
\begin{equation*}
  \mathsf{Enc}:\Sigma_q^k\to\Sigma_q^n,
  \qquad
  \mathsf{Dec}:\cF_{n,t}\to\Sigma_q^k,
\end{equation*}
such that, for every~$\bfx\in\Sigma_q^k$ and every break pattern
$\cP\subseteq[n-1]$ with~$|\cP|\le t$,
\begin{equation}\label{eq:brc-correctness}
  \mathsf{Dec}
  \left(
    \mathsf{Frag}_{\cP}(\mathsf{Enc}(\bfx))
  \right)
  =
  \bfx .
\end{equation}
The associated codebook is
\begin{equation*}
  \cC_\brc=\mathsf{Enc}(\Sigma_q^k)\subseteq\Sigma_q^n.
\end{equation*}

\section{Converse}
We extend the redundancy lower bound established for binary break-resilient codes in~\cite[Theorem 3.6]{wang2026break} to BRCs over an arbitrary fixed alphabet~$\Sigma_q$.
Our proof inherits the concept of $t$-confusability from the binary case, but does not rely on the sphere-packing argument.

\begin{definition}[$t$-confusability]\label{def:confusability}
  Two words~$\bfx,\bfy\in\Sigma_q^n$ are~$t$-confusable if there exist break patterns~$\cP_0,\cP_1\subset[n-1]$ with~$|\cP_0|=|\cP_1|\leq t$ such that
  \begin{equation*}
    \mathsf{Frag}_{\cP_0}(\bfx)=\mathsf{Frag}_{\cP_1}(\bfy).
  \end{equation*}
\end{definition}

Clearly, an~$(n,t)$-break-resilient code cannot contain~$t$-confusable codewords.

\begin{definition}[$\bfn$-equivalence]

Let~$\bfn= (n_1,\ldots,n_s)$ be a sequence of positive integers satisfying~$n_1+\cdots+n_s=n$.
For any~$\bfx,\bfy\in\Sigma_q^n$, write their unique block decompositions induced by~$\bfn$ as
\begin{equation*}
  \bfx = \bfx^{(1)}\circ\bfx^{(2)}\circ\cdots\circ\bfx^{(s)},~\mbox{and}~
  \bfy = \bfy^{(1)}\circ\bfy^{(2)}\circ\cdots\circ\bfy^{(s)},
  ~\mbox{where~$|\bfx^{(i)}|=|\bfy^{(i)}|=n_i$ for every~$i\in[s]$}.
\end{equation*}
We say~$\bfx$ and~$\bfy$ are~$\bfn$-equivalent, i.e.,~$\bfx \sim_\bfn\bfy$, if~$\bfx^{(i)}$ is a cyclic shift of~$\bfy^{(i)}$ for every~$i\in[s]$.
\end{definition}

The next lemma connects the previous two definitions.  
\begin{lemma}\label{lemma:connect-confuse-eq}
  For a positive integer~$s$ satisfying~$2s-1\leq t$, and any~$\bfx,\bfy\in\Sigma_q^n$,
  \begin{equation*}
    \bfx \sim_\bfn\bfy\Longrightarrow~\mbox{$\bfx$ and~$\bfy$ are~$t$-confusable}.
  \end{equation*}
\end{lemma}
\begin{proof}
  Consider words~$\bfx,\bfy\in\Sigma_q^n$ where~$\bfx\sim_\bfn\bfy$.
  Since~$\bfx^{(i)}$ is a cyclic shift of~$\bfy^{(i)}$ for every~$i\in[s]$, they have decompositions
  \begin{equation*}
    \bfx^{(i)}=\bfu^{(i)}\circ\bfv^{(i)}~\mbox{and}~\bfy^{(i)}=\bfv^{(i)}\circ\bfu^{(i)}. 
  \end{equation*}
  The adversary may perform the following two-stage procedures on~$\bfx,\bfy$, respectively.
  The adversary first breaks~$\bfx$ and~$\bfy$ into blocks based on~$\bfn$.
  Then, for blocks~$\bfx^{(i)}$ and~$\bfy^{(i)}$, the adversary further breaks them into~$\bfu^{(i)}$ and~$\bfv^{(i)}$.
  The two procedures yield the same multiset of fragments
  \begin{equation*}
    \{\{  \bfu^{(1)}, \bfv^{(1)},\ldots, \bfu^{(s)}, \bfv^{(s)}\}\},
  \end{equation*}
  and the number of breaks utilized is
  \begin{equation*}
    2s-1 \leq t.
  \end{equation*}
  Therefore,~$\bfx$ and~$\bfy$ are~$t$-confusable.
\end{proof}

The above lemma allows us to bound the size of a break-resilient code from above by the simple pigeonhole principle.

\begin{lemma}
For each positive integer $m$, let~$\cG_m$ be the group of cyclic shifts acting on~$\Sigma_q^m$, and denote the number of orbits of this action by
\begin{equation*}
  N_q(m)\triangleq|\Sigma_q^m/\cG_m |
\end{equation*}
For any positive integers~$n_1,\ldots,n_s$ satisfying~$2s-1\leq t$ and~$n_1+\cdots+ n_s=n$, every~$(n,t)$-BRC~$\cC_\brc$ satisfies
\begin{equation*}
  |\cC_\brc|\leq\prod_{i=1}^s N_q(n_i).
\end{equation*}
\end{lemma}

\begin{proof}
  Assume, for the sake of contradiction, that there are more than~$\prod_{i=1}^sN_q(n_i)$ codewords in~$\cC_\brc$.
  Then, there must exist two distinct codewords~$\bfx,\bfy\in\cC_\brc$ such that each pair of corresponding blocks lies in the same orbit, i.e.,~$\bfx^{(i)}$ is a cyclic shift of~$\bfy^{(i)}$ for every~$i\in[s]$.
  This makes~$\bfx\sim_\bfn\bfy$, where~$\bfn=(n_1,\ldots,n_s)$, and hence they are~$t$-confusable by Lemma~\ref{lemma:connect-confuse-eq}, which cannot coexist in a break-resilient code~$\cC_\brc$, a contradiction.
\end{proof}
 
We next provide the lower bound on the redundancy of $q$-ary break-resilient codes.
\begin{theorem}\label{theorem:converse}
  For every fixed integer~$q\geq2$, every~$(n,t)$-BRC over~$\Sigma_q$ has redundancy
  \begin{equation*}
  \Omega(t\log_q\frac{n}{t}).
  \end{equation*}
\end{theorem}
\begin{proof}
  Let~$g_\ell\in\cG_m$ denote right cyclic shift by~$\ell\in[0,m-1]$, which is a permutation with~$d=\gcd(m,\ell)$ cycles.
  Hence,~$g_\ell$ fixes elements in~$\Sigma_q^m$ with the same symbol in each cycle, and their number is~$q^{\gcd(m,\ell)}$.
  Applying Burnside’s lemma, we have
  \begin{align}
    N_q(m) &= \frac{1}{m}\sum_{\ell=0}^{m-1}q^{\gcd(m,\ell)}
    =\frac{1}{m}(q^m+\sum_{\ell=1}^{m-1}q^{\gcd(m,\ell)})
    \overset{(a)}{\leq} \frac{1}{m}(q^m+(m-1)\cdot q^{\frac{m}{2}})
    =\frac{q^m}{m}(1+ (m-1)\cdot q^{-\frac{m}{2}})\nonumber\\
    &{\leq}\frac{q^m}{m}(1+ (m-1)\cdot 2^{-\frac{m}{2}})
    \overset{(b)}{\leq}  \frac{2q^m}{m},\label{eq:bound-on-nqm}
  \end{align}
  where~$(a)$ is due to the observation that~$\gcd(m,\ell)$ is less than or equal to~$\frac{m}{2}$.
  For~$(b)$, consider the function
  \begin{equation*}
      f(z)=1+(z-1)2^{-z/2},
  \end{equation*}
  which has maximum value~$1+\frac{\sqrt{2}}{e\ln 2}<2$.
  
  Let~$s=\ceil{t/2}$ and let~$n_1,\ldots,n_s$ be balanced, i.e.,
  \begin{equation}\label{eq:balanced}
    n_1+\cdots+n_s=n,~\mbox{where}~n_i\in\{\floor{\frac{n}{s}},\ceil{\frac{n}{s}}\},
  \end{equation}
  The redundancy is then
  \begin{align*}
    n-\log_q|\cC_\brc| &\geq  \sum_{i=1}^s n_i- \log_q\prod_{i=1}^sN_q(n_i)
    =\sum_{i=1}^s \log_q \frac{q^{n_i}}{N_q(n_i)}
    \overset{\eqref{eq:bound-on-nqm}}{\geq}\sum_{i=1}^s \log_q \frac{n_i}{2}
    \overset{\eqref{eq:balanced}}{\geq} s \cdot \log_q (\frac{1}{2}\floor{\frac{n}{s}})\\
    &= \ceil{\frac{t}{2}}\cdot\log_q(\frac{1}{2}\floor{\frac{n}{\ceil{t/2}}})=\Omega(t\log_q\frac{n}{t}).\qedhere
  \end{align*}
\end{proof}

\section{Code Construction}\label{section:construction}
In this section, we present our code construction for a fixed finite field~$\bbF_q$, where~$q\geq 2$ is a prime power.
Throughout, we assume that~$t\ge 1$ and~$k>t$, and define
\begin{equation}\label{eq:define-big-var}
  T= t+1, K= 2T+2,~\mbox{and}~N=K+t.
\end{equation}

\subsection{Preliminaries}\label{section:preliminary}
Recall that the adversary may make at most~$t$ breaks in the codeword, after which the decoder receives the resulting unordered multiset of at most~$t+1$ fragments.
Since the order of these fragments is lost, decoding must account for all strings that can be obtained by concatenating the fragments.
We formalize this collection of possible assemblies as follows.

\begin{definition}[$t$-break ball]\label{def:ball}
For an integer~$t\ge 0$ and a~$q$-ary string~$\bfx\in\Sigma_q^n$, the~$t$-break ball of~$\bfx$ is defined as
\begin{equation*}
    \cB_t(\bfx)=
  \left\{
    \bfx_{I_{\pi(0)}}\circ \bfx_{I_{\pi(1)}}\circ \cdots \circ \bfx_{I_{\pi(\tau)}}\mid
    0\le \tau\le t, (I_0,\ldots,I_\tau) \in\cI_\tau(n),\pi\in\operatorname{Sym}(\{0,\ldots,\tau\})
  \right\},
\end{equation*}
where~$\cI_\tau(n)$ denotes the set of all ordered tuples~$(I_0,\ldots,I_\tau)$ of consecutive nonempty intervals obtained by breaking~$[n]$ in~$\tau$ positions. 
\end{definition}

\begin{lemma}\label{lemma:ball-size}
  For every~$\bfx\in\Sigma_q^n$, the size of its~$t$-break ball is bounded by~$n^t\cdot(t+1)^2$.
\end{lemma}
\begin{proof}
  By Definition~\ref{def:ball}, the ball~$\cB_t(\bfx)$ is the set of all~$q$-ary strings obtainable by breaking~$\bfx$ in at most~$t$ positions and permuting the resulting substrings.
  Hence,
  \begin{equation*} 
  |\cB_t(\bfx)|\leq \sum_{\tau=0}^{t}\binom{n-1}{\tau}(\tau+1)! \leq \sum_{\tau=0}^{t} {(n-1)^\tau} (\tau+1)\leq n^t\cdot (t+1) ^2 .\qedhere
\end{equation*}
\end{proof}

Next, we briefly review the definition of~\emph{mutually uncorrelated} (MU) codes that will be used in the code construction.
A code~$\cC_{\muc}$ is \emph{mutually uncorrelated} (MU) if, for any two (not necessarily distinct) codewords~$\bfx,\bfy\in\cC_{\muc}$, no nonempty proper prefix~of~$\bfx$ is equal to a suffix of~$\bfy$.
Equivalently, no two codewords in~$\cC_{\muc}$ can overlap with each other at a nontrivial shift.

For a string~$\bfu$, let~$\mathrm{zr}(\bfu)$ denote the length of its longest run of zeros.
Define a collection of words
\begin{equation}\label{eq:define-c-mu}
  \cC_{\muc}
  =
  \left\{
    0^{\ceil{\log_q k}+2}\circ 1\circ\bfu\circ 1
    \mid
    \bfu\in\bbF_q^{\ceil{\log_q t}+4},~\mathrm{zr}(\bfu) < \ceil{\log_q k}+2
  \right\}.
\end{equation}
The next lemma verifies that this construction yields sufficiently many mutually uncorrelated words.

\begin{lemma}
  The set $\cC_\muc$ is a mutually uncorrelated code with code size~$|\cC_\muc|\geq N$.
\end{lemma}
\begin{proof}
  Assume for the sake of contradiction that there are (not necessarily distinct) words~$\bfm_1,\bfm_2\in\cC_\muc$ such that a proper suffix of~$\bfm_1$ and a proper prefix of~$\bfm_2$ overlap.
  Let~$d$ be the overlap length such that~$0<d<n_\muc$.
  \begin{enumerate}
    \item If~$d\leq \ceil{\log_q k}+2$, then the prefix is~$0^d$, whereas the suffix ends in~$1$, contradiction.
    \item If~$d>\ceil{\log_q k}+2$, then the suffix must begin with~$0^{\ceil{\log_q k}+2}$.
    However, no proper suffix of~$\bfm_1$ begins with~$0^{\ceil{\log_q k}+2}$: the initial run
    of~$\ceil{\log_q k}+2$ zeros starts only at the first coordinate, and the interior word of~$\bfm_1$ contains no run of~$\ceil{\log_q k}+2$ zeros, contradiction.
  \end{enumerate}


  It remains to prove the cardinality bound.
  Note that, since~$k>t,q\geq 2$, and~$\ceil{\log_q k}\geq 1$, we have
  \begin{equation*}
    q^{\ceil{\log_q k}+1} \geq  \ceil{\log_q k}+3 \geq  \ceil{\log_q t}+3.
  \end{equation*}
  Taking logarithms on both sides, we have
  \begin{equation*}
    \ceil{\log_q k}+1\geq \log_q({\ceil{\log_q t}+3})\Longrightarrow\ceil{\log_q k}+1\geq \ceil{\log_q(\ceil{\log_q t}+3)}
    \Longrightarrow \ceil{\log_q k}+2 >\ceil{\log_q (\ceil{\log_q t}+3)}.
  \end{equation*}
  Using the uniquely decodable run-length limited (RLL) code in \cite[Alg.~1]{levy2018mutually}, a string of length~$a$ can be encoded to a run-length limited string of length~$a+1$ that is free of zero runs longer than~$\ceil{\log_q a}$.
  Hence, for every~$\bfv\in\bbF_q^{\ceil{\log_q t}+3}$, we can generate a constrained string~$\bfu\in\bbF_q^{\ceil{\log_q t}+4}$ such that~$\mathrm{zr}(\bfu)\leq\ceil{\log_q(\ceil{\log_q t}+3)}<\ceil{\log_q k}+2$.
  Therefore,
  \begin{equation*}
    |\cC_\muc| \geq q^{\ceil{\log_q t}+3}\geq q^3\cdot t\geq 8t\geq 3t+4 =N.\qedhere
  \end{equation*}
\end{proof}

The above lemma allows us to define~\emph{markers}, which are~$N$ distinct MU codewords
\begin{equation}\label{eq:define-markers}
  \bfm_0,\bfm_1,\ldots,\bfm_{N-1}\in\cC_\muc\subset  \bbF_q^{n_\muc},~\mbox{where}~n_\muc=\ceil{\log_q k}+\ceil{\log_q t}+8.
\end{equation}

\subsection{Encoding}\label{section:encoding}
The encoder maps an information word~$\bfx\in\bbF_q^k$ to an~$(n,t)$-break-resilient codeword~$\bfc\in\bbF_q^n$.
Specifically, the information word~$\bfx\in\bbF_q^k$ is first mapped to a marker-free~$q$-ary string~$\bfy\in\bbF_q^{k+1}$ of length~$k+1$, to which a~\emph{sketch} is then prepended to construct the final output codeword~$\bfc$.

The marker-removal transform draws heavily on the techniques presented in
\cite[Alg.~1]{levy2018mutually} and \cite[Alg.~1]{elishco2021repeat}, which iteratively replace a marker from the original string with its identity information and positional information.
To streamline the presentation, the transform, as well as its inverse, is presented in Appendix~\ref{section:marker-removal}.

\begin{remark}
  Note that the marker-removal transform requires the markers to have minimum length
   \begin{equation*}
    \ceil{\log_q k}+\ceil{\log_q (N)}+1= \ceil{\log_q k}+\ceil{\log_q (3t+4)}+1 \leq\ceil{\log_q k}+\ceil{\log_q t + \log_q 7}+1 < \ceil{\log_q k}+\ceil{\log_q t}+5\leq n_\muc,
  \end{equation*}
  which is satisfied by the choice of~$n_\muc$ in~\eqref{eq:define-markers}.
  It also requires all markers to be all~$0$-prefixed, which is also true in~\eqref{eq:define-c-mu}.
\end{remark}

Let~$Q$ be a power of~$q$ such that
\begin{equation}\label{eq:define-q}
   k\cdot |\cB_T(\bfy)|\leq  k\cdot (k+1)^{T}\cdot (T+1)^2 \leq Q <q\cdot k\cdot (k+1)^{T}\cdot (T+1)^2 .
\end{equation}
For the constrained string~$\bfy=y_0y_1\cdots y_k\in\bbF_q^{k+1}$, define the polynomial
\begin{equation}
  p_\bfy(z) = \sum_{i=0}^k y_iz^i\in\bbF_Q[z].
\end{equation}
We fix an arbitrary ordering of the elements of~$\bbF_Q$.
Following this ordering, the encoder chooses the first
$\alpha\in\bbF_Q$ such that
\begin{equation}\label{eq:define-alpha}
  p_\bfy(\alpha)\neq p_{\bfy'}(\alpha) \mbox{ for every }\bfy'\in\cB_T(\bfy)\setminus\{\bfy\}.
\end{equation}
Let~$\rep_q(\beta)\in\bbF_q^{\log_qQ}$ denote the $q$-ary representation of~$\beta\in\bbF_Q$.
The~\emph{fingerprint} of~$\bfy$ is then defined as
\begin{equation}\label{eq:define-fingerprint}
  \mathtt{fp}_\bfy=\rep_q(\alpha)\circ \rep_q(p_\bfy(\alpha))\in\bbF_q^{2\log_q Q}.
\end{equation}

\begin{lemma}
  There exists such an~$\alpha\in\bbF_Q$.
\end{lemma}
\begin{proof}
  The non-zero polynomial~$p_\bfy(z)-p_{\bfy'}(z)$ is of degree at most~$k$, and has at most~$k$ distinct roots.
  Across all~$\bfy'\in \cB_T(\bfy)\setminus\{\bfy\}$, the number of such roots is at most~$k\cdot (|\cB_T(\bfy)|-1)$.
  Due to~\eqref{eq:define-q}, such~$\alpha$ exists.
\end{proof}

The encoder slices~$\mathtt{fp}_\bfy$ into~$K$ chunks, each of length
\begin{equation*}
  \ceil{\frac{2\log_q Q}{K}}\leq \ceil{ 2\cdot\frac{1+(T+1)\log_q(k+1)+2\log_q(T+1)}{2(T+1)} }\leq\ceil{ \log_q(k+1)}+\ceil{\frac{1+2\log_q(T+1)}{(T+1)} }\leq  \ceil{\log_q k} +3.
\end{equation*}
The encoder then pads each chunk with zeros and treats them as field elements in~$\bbF_{q^{\ceil{\log_q k}+3}}$, and generates~$N$~\emph{MDS blocks}
\begin{equation*}
  \bfb_0,\bfb_1,\ldots,\bfb_{N-1}\in\bbF_q^{\ceil{\log_q k}+3}
\end{equation*}
using a~$[N,K]$ MDS code such that any~$K$ MDS blocks recover~$\mathtt{fp}_\bfy$.
The output codeword is then
\begin{equation}\label{eq:define-codeword}
  \bfc=\bfs\circ\bfy,~\mbox{where}~\bfs=\bfb_0\circ\bfm_0\circ\bfb_1\circ\bfm_1\circ\cdots\circ\bfb_{N-1}\circ\bfm_{N-1}.
\end{equation}

\begin{remark}
  With an~$[N,K]$ Reed--Solomon code, such encoding is possible since there are more than~$N$ field elements, i.e.,
  \begin{equation*}
    q^{\ceil{\log_q k}+3}\geq 8k > 3t+4 = N.
  \end{equation*}
\end{remark}
\begin{remark}
 The fingerprint defined in
\eqref{eq:define-alpha} is precisely the witness color produced by one round of cover-free recoloring~\cite{li2026constructing}, and $\mathtt{fp}_{\bfy}$ is its~$q$-ary representation.
\end{remark}

\subsection{Decoding}

Let~$\bfc=\bfs\circ\bfy$ be the codeword generated from the message~$\bfx$ using the procedure described in Section~\ref{section:encoding}.
Decoding begins by identifying markers from fragments and recovering~$\mathtt{fp}_\bfy$.
\begin{lemma}
  The decoder is guaranteed to recover~$\mathtt{fp}_\bfy$.
\end{lemma}
\begin{proof}

We first show that the decoder can unambiguously locate every marker that survives breaks.
Suppose a length-$n_\muc$ substring of the codeword equals a marker.
Since~$\bfy$ is marker-free, this substring must not reside entirely in~$\bfy$ and must begin in the sketch region.
Recall that the marker length~$n_\muc$ is greater than the MDS-block length, and hence the substring cannot reside entirely in one MDS block.
Consequently, it overlaps an actual marker.
If its starting position differs from the starting position of that marker, their nonempty overlap contradicts the MU property.
Therefore, every identified occurrence of a marker begins at the true boundary of an actual marker embedded in the codeword by the encoder.

Since~$t$ breaks fall into at most~$t$ blocks--marker units (i.e., a segment of the form~$\bfb_i\circ\bfm_i$), there exist at least~$N-t=K$ units whose symbols occur contiguously in one received fragment that survived breaks.
They can be located by the decoder sliding a window across every received fragment.
With the~$K$ units, the decoder extracts the MDS blocks, and obtains~$\mathtt{fp}_\bfy$ by decoding the~$[N,K]$ MDS code.
\end{proof}

The decoder learns~$\alpha$ and~$\beta=p_\bfy(\alpha)$, i.e., the evaluation of the polynomial~$p_\bfy$ at~$\alpha$, from the recovered~$\mathtt{fp}_\bfy$.
Then, for all possible assemblies of the received fragments whose prefix equals~$\bfs$, define~$\cW$ as the collection of their suffixes, which serve as candidates for the true~$\bfy$.
We now show that~$\bfy$ is the only candidate that satisfies the condition~$\beta=p_\bfy(\alpha)$.

First, Theorem~\ref{theorem:w-in-ball} shows that every candidate~$\bfw\in\cW\setminus\{\bfy\}$ belongs to the~$(t+1)$-break ball of~$\bfy$.
The proof is deferred to Appendix~\ref{section:proof-w-in-ball} to streamline the flow of presentation.

\begin{theorem}\label{theorem:w-in-ball}
    Let~$\bfc=\bfs\circ\bfy$ be a~$q$-ary string, and let~$\bfc'=\bfs\circ\bfw\in\cB_t(\bfc)$.
    Then~$\bfw\in\cB_{t+1}(\bfy)$.
\end{theorem}

Then, by the definition of~$\alpha$, the uniqueness of~$\bfy$ follows immediately.

\begin{corollary}\label{corollary:uniqueness}
  For every~$\bfw\in\cW\setminus\{\bfy\}$,
  \begin{equation*}
  \beta=p_\bfy(\alpha) \neq p_\bfw(\alpha).
  \end{equation*}
\end{corollary}

The decoder examines every candidate~$\bfw\in\cW$, evaluating~$p_\bfw(\alpha)$ until it finds~$\bfy$.
Finally, it outputs the~$\bfx$ by inverting the marker-removal transform, which concludes the decoding procedure.

\section{Analysis}
In this section, we analyze the redundancy and computational complexity of the construction presented in Section~\ref{section:construction}.
\begin{theorem}[Redundancy]
  The proposed~$(n,t)$-BRC has redundancy~$O(t\log_q n)$.
\end{theorem}
\begin{proof}
  The redundancy can be computed as
  \begin{equation*}
  \begin{split}
      n - k &\overset{\eqref{eq:define-codeword}}{=} N\cdot (\ceil{\log_q k}+3+\ceil{\log_q k}+\ceil{\log_q t}+8) + k + 1 -k=(3t+4)\cdot (2\ceil{\log_q k}+\ceil{\log_q t}+11)+1\\
      &\overset{\eqref{eq:define-big-var}}{=} (6t+8)\ceil{\log_q k}+(3t+4)\cdot (\ceil{\log_q t}+11)+1 = O(t\log_q n).\qedhere
  \end{split}
  \end{equation*}
\end{proof}

\begin{theorem}[Encoding complexity]\label{theorem:encoding-complexity}
  The proposed~$(n,t)$-BRC has encoding complexity~$O(t^6n^{2t+4}\log_q^2n)$.
\end{theorem}

\begin{proof}
Let~$B=(k+1)^{t+1}(t+2)^2$.
Recall that~$|\cB_T(\bfy)|\leq B$ by Lemma~\ref{lemma:ball-size}, and the choice of $Q$ in~\eqref{eq:define-q} satisfies
~$kB\leq Q<qkB$.
In the worst case, the unlucky encoder tests all $Q$ elements in
$\bbF_Q$.
For each element, it evaluates at most $B$ polynomials of degree at most $k$.
Horner's rule uses $O(k)$ operations in $\bbF_Q$ per evaluation.
Hence the fingerprint search requires
\begin{equation*}
  O(QBk)=O(k^2B^2)= O\left(k^2\cdot (k+1)^{2t+2}(t+2)^4\right)=O(t^4n^{2t+4})
\end{equation*}
operations in $\bbF_Q$.
Using schoolbook arithmetic, this equals~$O(t^4n^{2t+4}\log_q^2 Q)=O(t^6n^{2t+4}\log_q^2 n)$ operations in~$\bbF_q$.

Every marker replacement shortens the current string, so there are at
most $k+1$ replacements. A direct scan compares $N$ markers of length
$n_{\muc}$ at $O(k)$ positions, costing
$O(Nkn_{\muc})$ operations in~$\bbF_q$ per replacement.
Marker removal therefore costs
\begin{equation*}
   O(Nk^2n_{\muc})=O(tk^2(\log_q k +\log_q t))=O(tk^2\log_q n)
\end{equation*}
operations in~$\bbF_q$.
Finally, Reed--Solomon encoding uses $O(NK)$ operations in~$\bbF_{q^{\ceil{\log_q k}+3}}$ and hence 
\begin{equation*}
  O(NK \log_q^2n)= O(t^2\log_q^2n).
\end{equation*}
operations in~$\bbF_q$.
Constructing and writing the final codeword costs an additional $O(n)$ operations in~$\bbF_q$.
Combining these bounds proves the general complexity claim.
\end{proof}

\begin{theorem}[Decoding complexity]\label{theorem:decoding-complexity}
  The proposed~$(n,t)$-BRC has decoding complexity~$O\left((t+1)!\cdot n \cdot t^2\log_q^2n\right)$.
\end{theorem}
\begin{proof}
  Sliding the marker windows over all received fragments and directly
  comparing them with the $N$ markers costs
  \begin{equation*}
    O(nNn_{\muc})=O(n\cdot t\log_q n)
  \end{equation*}
  operations in~$\bbF_q$.
  Recovering the fingerprint from $K$ surviving MDS blocks requires~$O(K^2)=O(t^2)$ field operations in $\bbF_{q^{\ceil{\log_q k}+3}}$, which is
  \begin{equation*}
    O(t^2 \log_q^2n)
  \end{equation*}
  operations in~$\bbF_q$.
  Note that there are at most~$(t+1)!$ assemblies of the received fragments.
  Constructing an assembly and checking its prefix costs $O(n)$ operations in~$\bbF_q$.
  Evaluating the candidate polynomial at~$\alpha$ takes~$O(k)$ operations in~$\bbF_Q$.
  Therefore, testing all candidate assemblies costs
  \begin{equation*}
    O((t+1)! (n + kt^2 \log_q^2k))
  \end{equation*}
  operations in~$\bbF_q$.
  The inverse marker-removal transform has at most $k+1$ replacement steps.
  Using a data structure that supports insertion in~$O(\log_2 k)$ time to store the sequence, the inverse marker-removal can be implemented using $O(k\log_2^2n)$ operations in~$\bbF_q$.
  Combining the preceding bounds gives complexity
  \begin{equation*}
    O\left((t+1)!\cdot n\cdot t^2\log_q^2n\right).
  \end{equation*}
\end{proof}

\section{Conclusion}
In this paper, we established a redundancy lower bound of~$\Omega(t\log_q(n/t))$ for~$q$-ary break-resilient codes.
We also presented a deterministic~$q$-ary break-resilient code with redundancy
\begin{equation*}
  (6t+8)\ceil{\log_q k}+(3t+4)\cdot (\ceil{\log_q t}+11)+1 = O(t\log_q n).
\end{equation*}
For every~$t\leq n^{1-\varepsilon}$, this matches the information-theoretic lower bound up to a constant factor.

\section{Acknowledgment}
We thank Dr. Jin Sima for helpful discussions on extending the binary construction to the~$q$-ary setting.

\appendices
\renewcommand{\thetheorem}{\thesection.\arabic{theorem}}

\section{Marker-Removal Transform}\label{section:marker-removal}

For fixed integers~$l,k>0$, denote a set of~$0$-prefixed~\emph{markers} by
\begin{equation*}
  \bfm_0,\ldots,\bfm_{l-1}\in 0\circ\Sigma_q^{m-1},~\mbox{where}~m>\ceil{\log_q k}+\ceil{\log_q l}+1.
\end{equation*}
We describe a transform from an arbitrary message~$\bfx\in\Sigma_q^k$ to a marker-free string~$\bfy\in\Sigma_q^{k+1}$, as well as its inverse.

\subsection{Transform}
The encoder first appends a sentinel symbol~$1$ to~$\bfx$ and initializes
\begin{equation*}
  \bfz=\bfx\circ 1.
\end{equation*}
While~$\bfz$ contains a marker, let the leftmost occurrence be~$\bfm_j$, and write~$\bfz$ as
\begin{equation*}
  \bfz=\bfz'\circ\bfm_j\circ\bfz''~\mbox{where}~|\bfz''|=d.
\end{equation*}
Let~$\rep_{{\mathrm{id}}}(j)\in\Sigma_q^{\ceil{\log_q l}}$ and~$\rep_{{\mathrm{pos}}}(d)\in\Sigma_q^{\ceil{\log_q k}}$ denote the~$q$-ary representation of marker identity and right offset of~$\bfm_j$, respectively.
Note that~$d$ is the distance of the deleted marker from the
right end, rather than its absolute position in~$\bfz$.

During each~\emph{step}, the encoder deletes~$\bfm_j$ from~$\bfz$, and appends a~\emph{pointer}, defined as the identity and position information of~$\bfm_j$ followed by a sentinel symbol~$0$.
Specifically,~$\bfz$ is updated as
\begin{equation*}
  \bfz\gets\bfz'\circ\bfz''
  \circ\rep_{{\mathrm{id}}}(j)
  \circ\rep_{{\mathrm{pos}}}(|\bfz''|)
  \circ 0.
\end{equation*}
The encoder repeats this step until~$\bfz$ is marker-free; the termination is guaranteed in the following lemma.

\begin{lemma}
  The marker-removal process must terminate with a marker-free string~$\bfz$. 
\end{lemma}
\begin{proof}
  Since the length of a pointer is
  \begin{equation*}
    \ceil{\log_q l}+\ceil{\log_q k}+1<m,
  \end{equation*}
  replacing a marker with a pointer shortens the string~$\bfz$.
  Therefore, it is impossible for the encoder to enter infinite loop while~$\bfz$ remain unchanged after each replacement step, and hence termination is guaranteed.
\end{proof}

At termination, the resulting~$\bfz$ is marker-free, but may be shorter than~$k+1$.
The encoder then prepends a run of $1$'s to it, and outputs
\begin{equation*}
  \bfy =1^{k+1-|\bfz|}\circ\bfz \in \Sigma_q^{k+1}.
\end{equation*}
Since every marker is~$0$-prefixed, prepending a run of
$1$'s does not introduce a new marker occurrence.
Therefore,
\begin{theorem}
  The output~$\bfy\in\Sigma_q^{k+1}$ is marker-free.
\end{theorem}

\subsection{Inverse}

The inverse mapping recovers~$\bfx$ from the marker-free~$\bfy$ by reversing the aforementioned replacement steps in the opposite order.
The decoder initializes~$\bfz=\bfy$.
The last symbol of~$\bfz$ informs the decoder whether the suffix of~$\bfz$ is a pointer.
While it is~$0$, the decoder learns from the suffix the marker identity~$\bfm_j$ and marker position~$d$.
Recall that~$d$ is the distance of~$\bfm_j$ from the right end; this design enables the decoder to insert the marker without knowing the number of padded~$1$'s on the left of~$\bfy$.

The decoder then removes the suffix pointer from~$\bfz$, and obtains the intermediate string
\begin{equation*}
  \bfz'\circ\bfz'',~\mbox{where}~|\bfz''|=d.
\end{equation*} 
It recovers the deleted marker by setting
\begin{equation*}
  \bfz\gets \bfz'\circ\bfm_j\circ\bfz''.
\end{equation*}

Once the last symbol of~$\bfz$ is~$1$, all replacement steps have been reversed, and the current string is
\begin{equation*}
    \bfz=1^{|\bfz|-k-1}\circ\bfx\circ 1.
\end{equation*}
The decoder removes the final sentinel symbol~$1$ and iteratively removes the leading~$1$'s until the string length is~$k$, thereby recovering the original information word $\bfx$.
Hence,
\begin{theorem}
    The decoder outputs the unique~$\bfx\in\Sigma_q^k$ from which~$\bfy\in\Sigma_q^{k+1}$ was encoded.
\end{theorem}

\section{Proof of Theorem~\ref{theorem:w-in-ball}}\label{section:proof-w-in-ball}

\begin{figure}[t]
	\centering
	\begin{subfigure}{0.68\textwidth}
		\includegraphics[width=1\textwidth]{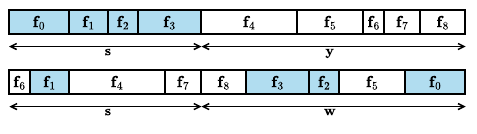}
		\caption{}
	\end{subfigure}\\
	\begin{subfigure}{0.68\textwidth}
		\includegraphics[width=1\textwidth]{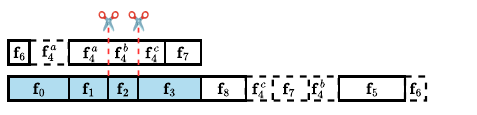}
		\caption{}
	\end{subfigure}\\
	\begin{subfigure}{0.68\textwidth}
		\includegraphics[width=1\textwidth]{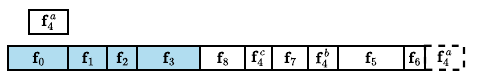}
		\caption{}
	\end{subfigure}\\
	\begin{subfigure}{0.68\textwidth}
		\includegraphics[width=1\textwidth]{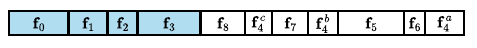}
		\caption{}
	\end{subfigure}
 	\caption{(a) Illustration of the true codeword~$\bfc=\bfs\circ\bfy$ and a candidate word~$\bfc'=\bfs\circ\bfw$, which are assembled from the same fragment multiset and share prefix~$\bfs$. The prefix fragments in~$\bfc$ are colored.
  (b) Let~$\varphi$ denote the permutation of coordinates from~$\bfc$ to~$\bfc'$, and let~$\psi$ be as defined in~\eqref{eq:define-psi}.
  We permute the fragments in~$\bfc$ using~$\psi$.
  The prefix fragments ($\bff_0,\ldots,\bff_3$) are fixed, and the suffix fragments are repeatedly permuted by~$\varphi$. Notably, after the first application of~$\varphi$, the fragment~$\bff_4$ crosses two internal boundaries of prefix fragments, and is split into~$\bff_4^a,\bff_4^b,\bff_4^c$.
  (c)-(d) The iteration continues until all suffix fragments reach the suffix region.
  }
  \label{fig:uniqueness}
\end{figure}

  Let~$r=|\bfs|$ and~$m=|\bfy|=|\bfw|$.
  We first split the fragment that crosses the boundary between the prefix~$\bfs$ and the suffix~$\bfy$ in~$\bfc$, and the fragment that crosses the boundary between the prefix~$\bfs$ and suffix~$\bfw$ in~$\bfc'$, if it exists.
  Now, every fragment either resides entirely in the prefix or entirely in the suffix in~$\bfc$.
  Meanwhile, every fragment either resides entirely in the prefix or entirely in the suffix, in~$\bfc'$.
  See Figure~\ref{fig:uniqueness}(a) for an illustrative example.

  The split process makes at most~$2$ extra breaks to the fragments.
  Let~$p$ denote the number of prefix fragments in~$\bfc$, and let~$s$ denote the number of suffix fragments in~$\bfc$, and we have
  \begin{equation}\label{eq:bound-fragments}
    p+s\leq t+3.
  \end{equation}

  Since~$\bfc'$ is a permuted assembly of the fragments that partition~$\bfc$, it naturally defines a permutation~$\varphi:[n]\to[n]$, where~$\varphi(i)$ is the coordinate in~$\bfc'$ occupied by the symbol which originally occupied coordinate~$i$ in~$\bfc$.
  Formally,
  \begin{equation}\label{eq:j-phi-j}
    c_j=c'_{\varphi(j)}~\mbox{for}~j\in[n].
  \end{equation}
  Note that this permutation~$\varphi$ is~\emph{order-preserving} within the range of a fragment.
  That is, for coordinates~$i,j$ such that the symbols~$c_i,c_j$ lie in one fragment,
  \begin{equation*}
    i<j\Longrightarrow \varphi(i)<\varphi(j).
  \end{equation*}
  Moreover, since~$\bfc$ and~$\bfc'$ share the common prefix~$\bfs$, we have
  \begin{equation}\label{eq:common-prefix}
    c_j=c'_{j}~\mbox{for}~j\in[r].
  \end{equation}

  Let~$\cP=[r]$,~$\cS=\{r+1,\ldots,n\}$, and for every~$j\in\cS$, define the first-return time
  \begin{equation*}
    u(j)=\min \{l>0\mid \varphi^{l}(j)\in \cS\}.
  \end{equation*}
  It is well-defined since the cycle of~$\varphi$ containing~$j$ eventually returns to~$j$.
  Equivalently,~$u(j)$ is the distance between coordinate~$j\in\cS$ and the next coordinate that is also in~$\cS$ within this cycle of~$\varphi$.
  Define
  \begin{equation}\label{eq:define-psi}
    \psi(j)=
    \begin{cases}
      j, &j\in\cP,\\
      \varphi^{u(j)}(j), &j\in\cS.
    \end{cases}
  \end{equation}
  We first verify that~$\psi$ is a permutation in the following lemma.
  \begin{lemma}\label{lemma:is-permutation}
    The map~$\psi$ is a permutation of~$[n]$.
  \end{lemma}
  \begin{proof}
    Assume that there exist distinct~$i,j\in[n]$ such that~$\psi(i)=\psi(j)$.
    Then,~$i,j\in\cS$ by the definition of~$\psi$.
    Without loss of generality, let~$u(i)\leq u(j)$, then
    \begin{equation*}
      \psi(i)=\varphi^{u(i)}(i)=\varphi^{u(j)}(j)= \psi(j)
      \Longrightarrow
      i =\varphi^{u(j)-u(i)}(j)\in\cS.
    \end{equation*}
    If~$u(i)=u(j)$, then~$i=j$, contradicting the assumption that~$i\neq j$.
    Otherwise,~$0<u(j)-u(i)<u(j)$, contradicting the definition of~$u(j)$ because~$i=\varphi^{u(j)-u(i)}(j)\in\cS$.
    Hence,~$\psi$ is injective and therefore a permutation of~$[n]$.
  \end{proof}

  Lemma~\ref{lemma:is-permutation} allows us to permute symbols in~$\bfc$ with~$\psi$, and leads to the following lemma.
  \begin{lemma}
    Let~$\bfd\in\Sigma_q^{r+m}$ be defined by permuting every symbol in~$\bfc$ using~$\psi$, i.e.,~$d_{\psi(j)}=c_j$.
    Then,~$\bfd=\bfc'$. 
  \end{lemma}
  \begin{proof}
    For~$j\in\cP$, this follows from~$\psi(j)=j$ and~\eqref{eq:common-prefix}.
    For~$j\in\cS$, by the definition of~$u(j)$, every intermediate value
    \begin{equation*}
      j_v=\varphi^v(j)\in\cP~\mbox{for}~1\leq v <u(j).
    \end{equation*}
    Therefore,
    \begin{equation*}
      c_{j_{v-1}}\overset{\eqref{eq:j-phi-j}}{=}c'_{j_v}\overset{\eqref{eq:common-prefix}}{=}c_{j_v},
    \end{equation*}
    and at termination,
    \begin{equation*}
      d_{\psi(j)}=c_j=\cdots = c_{j_{u(j)-1}}= c'_{j_{u(j)}}=c'_{\psi(j)}~\mbox{for}~j\in\cS.\qedhere
    \end{equation*}
  \end{proof}

  We now describe a procedure of constructing~$\bfc'=\bfd$ by permuting not the individual symbols, but fragments in~$\bfc$ with~$\psi$.
  Starting from prefix fragments, their positions remain unchanged and they partition the prefix region.
  Then, for every suffix fragment, we iteratively permute it using~$\varphi$.
  During each application of~$\varphi$, if the image in~$\cP$ crosses an internal boundary between two prefix fragments, split the fragment at the boundary and continue with the resulting pieces separately.
  The iteration continues until every piece reaches~$\cS$.

  Although the fragments are repeatedly being split during the procedure, the number of such splits is bounded.

  \begin{lemma}
    A boundary is used to split a fragment at most once.
  \end{lemma}
  \begin{proof}
    Define the first-return path of every coordinate~$j\in\cS$ as
    \begin{equation}
      \mathrm{Path}(j)=\left\{ \varphi^\ell(j)\mid 0<\ell<u(j)\right\}.
    \end{equation}
    We show that these paths are pairwise-disjoint.
    Assume for the sake of contradiction that there exist distinct coordinates~$i,j\in\cS$ and~$l\in\cP$ such that
    \begin{equation*}
      l\in \mathrm{Path}(i)\cap\mathrm{Path}(j)\neq\emptyset.
    \end{equation*}
    Then, applying~$\varphi$ to~$l$ until it reaches~$\cS$; the destination equals both~$\psi(i)$ and~$\psi(j)$ by definition in~\eqref{eq:define-psi}, which contradicts the fact that~$i\neq j$ due to the injectivity of~$\psi$.
    
    Finally, if the same internal boundary between prefix fragments is utilized twice, the coordinate on its immediate left would occur in two such paths, contradicting the pairwise-disjointness of paths.
  \end{proof}

  Since the procedure has introduced at most~$p-1$ splits to the suffix fragments in~$\bfc$, there are at most
  \begin{equation*}
    s+p-1 \overset{\eqref{eq:bound-fragments}}{\leq} t+2
  \end{equation*}
  fragments in the suffix region after the procedure; they form a multiset that partitions both~$\bfy$ and~$\bfw$.
  Therefore,~$\bfw$ can be obtained by breaking~$\bfy$ at most~$t+1$ times, i.e.,~$\bfw\in\cB_{t+1}(\bfy)$.

\printbibliography

\end{document}